\documentclass[11pt,twoside,a4paper,reqno]{amsart}

\usepackage{amssymb}
\usepackage{amsmath,amscd}
\usepackage[initials,nobysame,numeric]{amsrefs}
\usepackage{amsmath, amssymb}
\usepackage{amsfonts}
\usepackage{mathrsfs}
\usepackage{mathpazo}
\usepackage[arrow,matrix,curve,cmtip,ps]{xy}
 \usepackage{paralist}

\usepackage{amsthm}

\usepackage{float}
\usepackage{hyperref}
\usepackage{graphics}
\usepackage{graphicx}
\usepackage{verbatim}
\usepackage{multirow}
\usepackage{tikz}
\usepackage{enumerate}
\usepackage{eufrak}
\usepackage{bbm}

\allowdisplaybreaks

\newtheorem{theorem}{Theorem}[section]

\newtheorem{lemma}[theorem]{Lemma}

\theoremstyle{definition}
\newtheorem{definition}[theorem]{Definition}
\newtheorem{example}[theorem]{Example}
\newtheorem{remark}{Remark}

\numberwithin{equation}{section}

\newcommand{\N}{\mathbb{N}}
\newcommand{\R}{\mathbb{R}}
\newcommand{\C}{\mathbb{C}}

\newcommand{\s}{\mathbb {S}}
\newcommand{\E}{\mathbb{E}}

\newcommand{\K}{\mathcal {K}}

\newcommand{\Ordo}{\mathcal {O}}

\begin{document}

\date{\today}

\address{D. Mustafa and B. Wennberg\\Department of Mathematics\\Chalmers University of Technology and the University of Gothenburg\\412 96 G\"OTEBORG\\SWEDEN}

\email{dawan@chalmers.se and wennberg@chalmers.se}




\author{D. Mustafa and B. Wennberg}

\title[Chaotic distributions for relativistic particles]{Chaotic distributions for relativistic particles}

\begin{abstract}
We study a modified Kac model where the classical kinetic energy is replaced
by an arbitrary energy function $\phi(v)$, $v \in \R$. The aim of this paper
is to show that the uniform density with respect to the microcanonical
measure is $Ce^{-z_0\phi(v)}$-chaotic, $C,z_0 \in \R_+$. The kinetic energy
for relativistic particles is a special case. A generalization to the case
$v\in \R^d$ which involves conservation momentum is also formally discussed.
\end{abstract}

\maketitle
\section{Introduction}
\label{sec:intro}%
In 1956, Mark Kac published a paper \cite{kac11}, in which he answered some
fundamental questions concerning the derivation of the spatially homogen\-eous
Boltzmann equation. Kac considered a stochastic model consisting of $N$
identical particles, and obtained an equation like the spatially homogeneous
Boltzmann equation as a mean-field limit when the number of particles tends to
infinity. The key ingredient in his paper was the notion of \emph{chaos}, which
goes back to Boltzmann's {\em stosszahlansatz}. Loosely speaking, chaos is
related to asymptotic independence. We give a precise mathematical definition
later. He showed that if the probability distribution of the initial state of
the particles is chaotic then this property is propagated in time, i.e., the
probability distribution at time $t>0$ is also chaotic. This is referred to as
propagation of chaos.

Kac's model is a jump process on $\s^{N-1}=\{ (v_1,...,v_N) \in \R^N\,|\,
v_1^2+\cdots +v_N^2=N\}$. This corresponds to an $N$-particle system where a
particle is represented by its velocity $v\in\R$, and the kinetic energy of a
particle is given by $v^2$. The total energy is conserved in this process,
which is ergodic and hence has a unique invariant measure. In the original
model, the stationary distribution is the uniform measure on $\s^{N-1}$. It is
easy to prove that this uniform measure is $(2\pi)^{-1/2} e^{-v^2/2}$-chaotic
according to the definition given below, and Kac also provides a method of
constructing other families of chaotic measures on $\s^{N-1}$ which may serve as
suitable initial distributions for the process.

In this paper we consider a generalization of Kac's model, where the
energy of a particle is given by a large class of functions $\phi(v)$,
which includes $\phi(v)=v^2$. The motivation for studying this problem
comes from the theory of special relativity, where $v$ represents the
momentum of a particle rather than its velocity, and where the kinetic
energy is given  by
\begin{equation*}
\phi(v)=\sqrt{1+v^2}-1.
\end{equation*}
The kinetic theory for relativistic particles is far less developed
than for classical particles, but there is a growing interest in this
kind of problems (see below for references), and this paper is a
contribution to this work.

The first result of this paper is a proof that the uniform distribution on the
mani\-fold represented by the total kinetic energy of the particles (depends on
$\phi$) with respect to the microcanonical measure is chaotic by using the
approach of Kac. We also consider a jump process on this manifold, which is a
direct generalization of Kac's model, and prove that propagation of chaos holds
for this process. We also consider a generalization to a model where
the particles have momenta $v\in \R^d$.

One of the most influential equations in the mathematical kinetic theory of
gases is the Boltzmann equation, which describes the time evolution of the
density of a single particle in a gas consisting of a large number of
particles.

The Boltzmann equation is

\begin{equation}\label{hbe}
\begin{cases}
\frac{\partial}{\partial t}f(x,v,t)+\tilde{v}\cdot\nabla_x f(x,v,t) =Q(f,f)(x,v,t), \hspace{0.1in}v\in \R^3,t>0, \\
f(v,0)=f_0,
\end{cases}
\end{equation}
where in the classical case
$\tilde{v}=v$, and in the relativistic case
$\tilde{v}=v/\sqrt{1+|v|^2}$. The collision
operator on the right hand side is a quadratic operator acting only in the velocity
variables. In the classical case it is given by
\begin{equation*}
Q(f,f)=\int_{\R^3}\int_{\s^2}\left( f(v',t)f(v'_*,t)-f(v,t)f(v_*,t)\right)B(v-v_*,\sigma)dv_*d\sigma.
\end{equation*}

The function $B$ is called the collision kernel and is derived from the physics of the model. The corresponding collision term is considerably more
complicated in the relativistic case, which is thoroughly
investigated in~\cite{StrainYun2014}, for example. In both cases,
$(v,v_*)$ and $(v',v_*')$ represent the velocities of a pair of
particles before and after an elastic collision, and satisfy the
conservation of momentum,
\begin{equation}
     v+v_* = v'+v_*'\,,
\end{equation}
and energy
\begin{equation}
  \phi(|v|) + \phi(|v_*|) =   \phi(|v'|) + \phi(|v_*'|)\,.
\end{equation}

The derivation of the Boltzmann equation from a classical (Newtonian) many
particle system is still an important research topic. The classical result in
this direction is by Lanford \cite{lanford111}, however, this result is valid
only over a fraction of the mean free time between collisions. Up to this date
the result by Lanford has only been improved upon only in details, see e.g.,
\cite{pss} and the recent book~\cite{Gallgher_etal2013}. For the study of the
relativistic Boltzmann equation, see e.g., \cite{relbook}.

To obtain an equation like (\ref{hbe}) as a mean-field limit, Kac
considered the \emph{master equation}, which describes the evolution of the probability density on state space.
Assuming that velocities of the particles is distributed according to a probability
density $F_0$ initially, its time evolution is given by the following
linear master equation
\begin{equation}\label{mastereq_Kac}
\begin{cases}
\frac{\partial}{\partial t}F_N(v_1,\dots,v_N,t)=\K F_N(v_1,\dots,v_N,t), \\
F_N(v_1,\dots,v_N,0)=F_0(v_1,\dots,v_N).
\end{cases}
\end{equation}
where $\K$ is the generator of a Markov jump process,
\begin{equation*}
\K F_N(v_1,\dots,v_N)=N [Q-I]F_N(v_1,\dots,v_N),
\end{equation*}
with $I$ being the identity operator and
\begin{equation*}
QF_N(v_1,\dots,v_N)=\frac{2}{N(N-1)}\sum_{i<j}\int_{-\pi}^{\pi}
  F_N(v_1,\dots,v_i'( \theta),\dots,v_j'(\theta),\dots,v_N)\frac{d\theta}{2\pi}.
\end{equation*}
Hence $\K$ can be written as a sum of generators acting on only two
variables, $v_i$ and $v_j$, say, which corresponds to  binary collisions in a
real gas. The pair of velocities $v_i'(\theta)$, $v_j'(\theta)$ is the outcome
of the velocities $v_i$, $v_j$ undergoing a collision, and in the classical
case, as proposed by Kac they are given by
\begin{equation}\label{colvelKac}
v_i'(\theta)=v_i\cos \theta + v_j \sin \theta , \hspace{0.2in} v_j'(\theta)=-v_i \sin \theta + v_j\cos \theta .
\end{equation}
The energy of a pair of particles is always conserved in a collision, i.e.,
$$v_i'(\theta)^2+v_j'(\theta)^2=v_i^2+v_j^2.$$
In a realistic model the momentum should also be conserved, but since the
particles have one-dimensional velocities, imposing conservation of momentum
would lead to either the particles keeping the velocities during a collision or
exchanging velocities. Momentum conservation is therefore sacrificed in this
case.

The fact that particles are assumed to be indistinguishable corresponds to the
initial distribution  $F_0$ being symmetric with respect to permutations, and
this property is preserved by $\mathcal{K}$, so that  $F_N$ is symmetric for
all $t>0$.

In order to obtain a mean-field limit of the master equation
(\ref{mastereq_Kac})  as $N\rightarrow \infty$, Kac introduced the notion of
chaos (in \cite{kac11}, he referred to it as the Boltzmann property).
\begin{definition}\label{def_chaos}
Let $f$ be probability density on $\R$. For each $N \in \N$, let $F_N$ be a
probability density on $\R^N$ with respect to a measure $m^{(N)}$. The sequence
$\{F_N\}_{N \in \N}$ of probability densities on $\R^N$ is said to be
$f-$chaotic if the following two conditions are satisfied:
\begin{enumerate}
\item Each $F_N$ is a symmetric function of $v_1,\dots,v_N$.
\item For each fixed $k\in \N$, the $k$-th marginal $f_k^N(v_1,\dots,v_k)$ of
    $F_N$ converges to $\prod_{i=1}^{k}f(v_i)$, as $N\rightarrow \infty$,
    where $f(v)=\lim_{N\rightarrow \infty}f_1^N(v)$. The convergence is to be
    taken in the weak sense, that is, if $\varphi(v_1,\dots,v_k)$ is a
    bounded continuous function of $k$ variables, $v_1,\dots,v_k\in \R$, then
\begin{equation}
\begin{split}
&\lim_{N\rightarrow \infty}\int_{\R^N}\varphi(v_1,\dots,v_k) F_N(v_1\dots,v_N) dm^{(N)}\\
&\hspace{2cm}=\int_{\R^k}\varphi(v_1,\dots,v_k) \prod_{i=1}^{k}f(v_i) dv_1 \dots dv_k.
\end{split}
\end{equation}
\end{enumerate}
\end{definition}
More generally, this can be formulated for probability densities on products
$E^N$ where $E$ may be an arbitrary Polish space. The setting of this paper is
also slightly different from the Definition 1.1 because the dynamics take place
on a submanifold of $R^N$.

A chaotic family of probability densities on $\s^{N-1}(\sqrt{N})$ is the following:
\begin{example}\label{exchaoticsph}
It is a well known fact that the surface area of the sphere $\s^{N-1}(\sqrt{E})$ in $\R^N$ is given by
\begin{equation*}
|\s^{N-1}(\sqrt{E})|=\frac{2 \pi^{\frac{N}{2}}E^{\frac{N-1}{2}}}{\Gamma \left(\frac{N}{2} \right )}.
\end{equation*}
Let
$$F_N(v_1,\dots,v_N)=\frac{1}{|\s^{N-1}(\sqrt{N})|}$$
be the symmetric uniform density on $\s^{N-1}(\sqrt{N})$ with respect to
surface measure $\sigma^{(N)}$. Let $\varphi$ be continuous function on $\R$.
It is easy to see that
\begin{equation*}
\int_{\s^{N-1}(\sqrt{E})} \varphi(v_1)d \sigma^{(E)}=\int_{-\sqrt{E}}^{\sqrt{E}}
      \varphi(v_1)\left |\s^{N-2}\left (\sqrt{E-v_1^2}\right )\right| d v_1.
\end{equation*}
Replacing $E$ with $N$ we have
\begin{equation*}
\begin{split}
 \displaystyle \lim_{N\rightarrow \infty}&\frac{\int_{\s^{N-1}(\sqrt{N})}\varphi(v_1)
  d \displaystyle \sigma^{(N)}}{|\s^{N-1}(\sqrt{N})|}\\
  &= \lim_{N\rightarrow \infty} \frac{\Gamma \left(\frac{N}{2} \right )}{\pi^ {\frac{1}{2}}N^{\frac{1}{2}} \Gamma \left(\frac{N-1}{2} \right )}
  \int_{-\sqrt{N}}^{\sqrt{N}}\varphi(v_1)\left (1- \frac{v_1^2}{N}\right )^{\frac{N-3}{2}}d v_1 \\
  &=\frac{1}{\sqrt{2\pi}}\int_{-\infty}^{\infty} \varphi(v_1) e^{-\frac{v_1^2}{2}} dv_1.
\end{split}
\end{equation*}
Taking $\varphi$ to be a function on  $\R^k$, the same calculations show that, the family of uniform densities
on $\s^{N-1}(\sqrt{N})$, with respect to the surface measure, is $\frac{1}{\sqrt{2\pi}}e^{-\frac{v^2}{2}}$ -chaotic.
\end{example}
Kac showed by using a combinatorial argument that the master equation (\ref{mastereq_Kac}) propagates chaos, that is, if $\{F_N(v_1,\dots,v_N,0)\}_{N \in \N}$ is
$f_0-$chaotic, then the solutions to (\ref{mastereq_Kac}), $\{F_N(v_1,\dots,v_N,t)\}_{N \in \N}$ is $f(v,t)-$chaotic. The density $f(v,t)$ satisfies the Boltzmann-Kac equation
\begin{equation}
\begin{split}
\frac{\partial}{\partial t}f(v,t)&=2\int_{\R}\int_{-\pi}^{\pi}\left[f(v(\theta),t)f(w(\theta),t)-f(v,t)f(w,t)\right]\frac{d\theta}{2\pi}dw, \\
f(v,0)&=f_0,
\end{split}
\end{equation}
with $v(\theta),w(\theta)$ given by (\ref{colvelKac}).

Kac described also a method to construct other chaotic probability
densities on $\s^{N-1}(\sqrt{N})$: Let
\begin{equation}\label{chaoticdist}
F_N(v_1,\dots,v_N)=\frac{\prod_{i=1}^{N} f(v_i)}{Z(\sqrt{N})},
\end{equation}
where
\begin{equation}
Z(\sqrt{E})=\int_{\s^{N-1}(\sqrt{E})}\prod_{i=1}^{N} f(v_i)d\sigma^{(E)},
\end{equation}
with $\sigma^{(E)}$ being the surface measure on $\s^{N-1}(\sqrt{E})$. Kac showed that,
under some conditions on the function $f$, the family of probability densities
$\{F_N\}_{N \in \N}$ is $h(v)-$chaotic, where
\begin{equation}
h(v)=\frac{e^{-z_0v^2}f(v)}{\int_{\R}e^{-z_0v^2}f(v)dv},
\end{equation}
and $z_0$ is a positive constant. Note that the case $f(v)=1$ corresponds to the previous example.

The purpose of the present article is to study chaotic probability densities on
other surfaces than $\s^{N-1}(\sqrt{N})$. Define
\begin{equation}\label{spaceOmega}
\Omega^{N-1}(\sqrt{E})=\left \{(v_1,\dots,v_N)\; | \; \sum_{i=1}^{N}\phi(v_i)=E \right \},
\end{equation}
where $\phi(v)$ represents the energy of a particle with velocity $v$, and
\begin{eqnarray}\label{phivillkor}
  \phi &\in& C^1(\R,\R^+),\qquad\qquad \phi \quad \mbox{is even and convex},\nonumber\\
  \phi(0)&=&0\,.
\end{eqnarray}
We show that the uniform density on $\Omega^{N-1}(\sqrt{N})$ with respect
to the microcanonical measure  is $Ce^{-z_0\phi(v)}-$chaotic, where $C$ is a
normalisation constant and $z_0>0$ is the unique solution to a specific
equation.

Chaotic probability densities on spheres have been considered
by many authors. In a paper by Carlen, Carvalho, Le Roux, Loss and Villani,
\cite{cclrv}, the authors show the following theorem, which differs from Kac
result both in the exact definition of chaos (the notion of entropic chaos is
introduced), and in the method of proof:
\begin{theorem}
Let $f$ be a probability density on $\R$ such that $f\in L^{p}(\R)$ for some
$p>1$, $\int_{\R}v^2f(v)dv=1$ and $\int_{\R}v^4f(v)dv<\infty$. Let a family of
densities $\{ F_N\}_{N\in \N}$, be defined with $F_N$ as in(\ref{chaoticdist})
with $\sigma^{(N)}$  now being the normalized surface measure on
$\s^{N-1}(\sqrt{N})$. Then  $\{ F_N\}_{N\in \N}$
 is $f$-chaotic.
\end{theorem}

Sznitman~\cite{sznitman111} constructed chaotic families of measures
using a somewhat different approach: Let $h_1,\dots,h_N$
be i.i.d random variables with law $\mu(dh)=f(h)dh$, where $h\in \R^d$. Assume
that $f$ is differentiable and satisfies the following condition
\begin{equation}\label{conditionflambda}
\int_{\R^d} (f(h)+|\nabla f(h)|)e^{\lambda |h|} dh < \infty,
\end{equation}
where $\lambda\in \R$. Then, the conditional distribution $\mu^N$ of
$(h_1,\dots, h_N)$ subject to the constraint $\frac{h_1+\dots+h_N}{N}=a$ is
$\Upsilon-$chaotic, where
\begin{equation}\label{eq:sz14}
\Upsilon =\frac{1}{Z_{\lambda_a}}e^{\lambda_a \cdot h}\mu(dh),
\end{equation}
 with $\lambda_a$ determined by the equation $\int h \, d\Upsilon
(h)=a$ and $Z_\lambda$ is a normalisation constant. Note that, the choice of
$\lambda$ for which condition (\ref{conditionflambda}) is satisfied depends on
$a$. Within the framework of this paper we think of the random variable $h\in
\R$ as the energy of a particle, in this case $h=\phi(v) \geq 0$.

The organization of this paper is as follows. In section $2$ we define the
\emph{micro\-canonical} measure on $\Omega^{N-1}(\sqrt{N})$ and show that the
uniform density on \\
$\Omega^{N-1}(\sqrt{N})$ with respect to this measure is chaotic. In
Section~$3$, we introduce a modified Kac model with particles having the energy
$\phi(v)$ and the corresponding master equation. We also show that this master
equation propagates chaos and obtain the limiting equation. In Section~$4$, we
discuss how to generalize the results in Section~$2$ to the case $v\in \R^d$,
$d>1$, in which case the momentum is also conserved. An Appendix is also
included to introduce the methods that are used in the paper.
\section{Chaotic measures}
\label{sec:chaotic}%
 This section is devoted to showing that the family of
uniform probability densities on $\Omega^{N-1}(\sqrt{N})$ with respect to
the microcanonical measure is chaotic. By Example \ref{exchaoticsph}, another
approach would be show that the uniform density on
$\Omega^{N-1}(\sqrt{N})$ with respect to surface measure is chaotic. However,
since we consider a particle system where the total energy is conserved it is
natural to use the microcanonical measure on $\Omega^{N-1}(\sqrt{N})$. We show
that the family of uniform probability densities on
$\Omega^{N-1}(\sqrt{N})$ with respect to the microcanonical measure is
$Ce^{-z_0\phi(v)}$-chaotic where $z_0>0$ is the solution to a specific equation
and $C$ is a normalisation constant.

The microcanonical measure on $\Omega^{N-1}(\sqrt{N})$ is defined as follows: Let
\begin{equation}
H(v_1,\dots,v_N)=\sum_{i=1}^{N}\phi(v_i).
\end{equation}
\begin{definition}
Provided that $|\nabla H| \neq 0$, the microcanonical measure, $\eta^{(E)}$, on $\Omega^{N-1}(\sqrt{E})$, is defined by
\begin{equation}
\eta^{(E)}=\frac{\sigma_{\Omega}}{|\nabla H|},
\end{equation}
where $\sigma_{\Omega}$ is the surface measure induced by the Euclidean measure in $\mathbb{R}^N$ on $\Omega^{N-1}(\sqrt{E})$.
\end{definition}

The microcanonical measure arises naturally in physics from the assumption of equal probability, meaning
essentially that all microscopic particle configurations corresponding to the same energy
are equally probable, see {\em e.g.}~\cite{lrs} or~\cite{TKS1992}.

A more geometrical approach is given  by the coarea formula,
as explained {\em e.g.} in~\cite{EvansGariepy2015}:
If $\Phi: \R^n\rightarrow\R^m$ is Lipschitz continuous and $n\ge m$, then for
each measurable set $A\subset \R^n$,
\begin{equation}
    \int_A J \Phi(x) dx = \int_{\R^m} \mathcal{H}^{n-m}( A\cap \Phi^{-1}(y))\,dy\,,
\end{equation}
where $J \Phi $ is the Jacobian of $\Phi$, as defined
in~\cite{EvansGariepy2015}, and $\mathcal{H}^{n-m}$ denotes the
$n-m$-dimensional Hausdorff measure. In our setting this can be rephrased as
(see Theorem 3.13 in~\cite{EvansGariepy2015}): if $H : \R^n\rightarrow\R$ is
Lipschitz continuous and $\mathrm{ess\; inf} |\nabla H|>0$, and
$g:R^n\rightarrow \R$ is integrable, then for almost all $E$,
\begin{equation}
  \frac{d}{d E} \int_{\R^N} g(x)\chi_{\{H(x)>E\}}\,dx =
     \int_{\{H=E\}}
   \frac{g}{|\nabla H|} d\mathcal{H}^{n-1}\,.
\end{equation}
where $\chi$ is the indicator function of a set. Because
$H(v_1,...,v_N)=\phi(v_1)+\cdots\phi(v_N)$ with $\phi$ convex and diffentiable,
the set $\{H=E\}$ is also convex and regular. Then
\begin{equation*}
\begin{split}
&\int_{\R^N}g(v_1,\dots,v_N) \delta(H(v_1,\dots,v_N)-E)dv_1 \dots dv_N \\
&=\int_{\R} \int_{\Omega^{N-1}(\sqrt{y})}g\delta(y-E)\frac{d\sigma_{\Omega}}{|\nabla H|}dy
=\int_{\Omega^{N-1}(\sqrt{E})}g\frac{d\sigma_{\Omega}}{|\nabla H|}.
\end{split}
\end{equation*}
This shows that
$$\delta(H(v_1,\dots,v_N)-E)=\frac{\sigma_{\Omega}}{|\nabla H|}.$$

Another concept that is relevant in this context is that of disintegration of a
measure, which is a measure theoretic approach to almost the same problem, and
a means of approaching conditional probabilities, see {\em
e.g.}~\cite{ChangPollard1997}. In our setting, the disintegration theorem
states that given a measure $d\mu$ on $\R^N$ there are a measure $\nu$ on
$\R^+$ and a family of measures $\mu_{E}$ on the level surfaces $\{
(v_1,...,v_N) \; | \; \sum \phi(v_j) = E\} $, such that
\begin{equation}
\int_{\R^N} f(y) d\mu(u) = \int_{0}^{\infty} \int_{\{ \sum \phi(v_j)=E\} } f(y) d\mu_{E}(y) d\nu(E)\,,
\end{equation}

\begin{remark}
On $\s^{N-1}(\sqrt{E})$, we have that $|\nabla H|=2\sqrt{E}$. This implies that
the microcanonical measure is up to a constant factor equal to the surface
measure on $\s^{N-1}(\sqrt{E})$.
\end{remark}
Using the equality $H(v_1, \dots ,v_N)=E$, we can express, at least locally, the variable $v_N$ as function of $v_1,\dots,v_{N-1}$:
$$v_N=U(v_1,\dots,v_{N-1}).$$
By this parametrization, the surface $\Omega^{N-1}(\sqrt{E})$ can be represented as the graph of $U:\R^{N-1}\rightarrow \R$. The surface measure $\sigma_{\Omega}$ on $\Omega^{N-1}(\sqrt{E})$ is now given by
$$d\sigma_{\Omega}=\sqrt{1+|\nabla U|^2} \ dv_1 \dots dv_{N-1}.$$
By the implicit function theorem it follows that
$$\frac{\partial U}{\partial v_k}=
-\frac{\frac{\partial H}{\partial v_k}}{\frac{\partial H}{\partial v_N}}\hspace{0.2in}k=1,\dots,N-1.$$
Thus
\begin{equation*}
\frac{d\sigma_{\Omega}}{|\nabla H|}=\frac{1}{\left |\frac{\partial H}{\partial v_N}\right |}dv_1 \dots dv_{N-1}.
\end{equation*}

To carry out integration on $\Omega^{N-1}(\sqrt{E})$ with respect to the microcanonical measure $\eta^{(E)}$ we use the last equality:
\begin{equation*}
\begin{split}
&\displaystyle \int_{\Omega^{N-1}(\sqrt{E})}g(v_1,\dots,v_N)d\eta^{(E)}\\
&\hspace{2cm}=\sum_{\epsilon =+,-}\int_{\sum_{i=1}^{N-1}\phi(v_i)\leq E}
      g(v_1,\dots,\epsilon v_N)
      \frac{1}{\left |\frac{\partial H}{\partial v_N}\right |_{\epsilon}}dv_1 \dots dv_{N-1},
\end{split}
\end{equation*}
where
\begin{equation*}
  v_N= \phi^{-1}\left(E-\sum_{i=1}^{N-1}\phi(v_i)\right),
\end{equation*}
and $\phi^{-1}(v)$ is the of $\phi(v)$, $v\geq 0$. Moreover,
\begin{equation*}
\left |\frac{\partial H}{\partial v_N}\right |_{\epsilon}=
\left |\frac{\partial H}{\partial v_N}\left(v_1,\dots,v_{N-1},\epsilon\phi^{-1}
  \left(E-\sum_{i=1}^{N-1}\phi(v_i)\right)\right)\right |.
\end{equation*}

The uniform density $F(v_1, \dots ,v_N)$ on $\Omega^{N-1}(\sqrt{N})$ with
respect to the microcanonical measure $\eta^{(N)}$ is given by
\begin{equation}\label{distonomega}
F_N(v_1,\dots,v_N)=\frac{1}{Z_\phi(\sqrt{N})},
\end{equation}
where
\begin{equation*}
Z_\phi(\sqrt{E})=\int_{\Omega^{N-1}(\sqrt{E})}d\eta(E) .
\end{equation*}
To show that the uniform density on $\Omega^{N-1}(\sqrt{N})$ with respect
to the microcanonical measure $\eta^{(N)}$ is $C e^{-z_0 \phi(v)}$-chaotic we
follow Kac \cite{kac11}, and start by determining the asymptotic behaviour of
\begin{equation*}
Z_\phi(\sqrt{E})=\displaystyle \sum_{\epsilon =+,-} \int_{\sum_{i=1}^{N-1}\phi(v_i)
  \leq E}\frac{1}{\left |\frac{\partial H}{\partial v_N}\right |_{\epsilon}} dv_1 \dots dv_{N-1}
\end{equation*}
with $E=N$ for large $N$. Since
\begin{equation*}
\left|\frac{\partial H}{\partial v_N}\right |_{\epsilon}=|\phi'(v_N)| \hspace{0.1in} \text{and}\hspace{0.1in}
v_N=\pm \phi^{-1}\left (E-\sum_{i=1}^{N-1}\phi(v_i) \right ),
\end{equation*}
we have
\begin{equation*}
Z_\phi(\sqrt{E})=\displaystyle 2 \int_{\sum_{i=1}^{N-1}\phi(v_i)\leq E} \frac{1}{\left |\phi'(\phi^{-1}(E-\sum_{i=1}^{N-1}\phi(v_i))) \right |}dv_1 \dots dv_{N-1}.
\end{equation*}
To write $Z_\phi(\sqrt{E})$ as an integral over the sphere $\s^{N-1}(\sqrt{E})$, we make the change of variables $y_i^2=\phi(v_i)$ with respect to sign of $v_i$, $i=1,\dots,N-1$. This leads to

\begin{equation}\label{ZNchangevar}
\begin{split}
Z_\phi&(\sqrt{E})=\\
&2^{N} \displaystyle \int_{\sum_{i=1}^{N-1}y_i^2 \leq E }
\frac{1}{\left |\phi'(\phi^{-1}(E-\sum_{i=1}^{N-1}y_i^2 ))\right |}
\prod_{i=1}^{N-1}\frac{|y_i|}{|\phi'(\phi^{-1}(y_i^2))|} dy_1 \dots dy_{N-1}.
\end{split}
\end{equation}
Let
\begin{equation}\label{functioncy}
f(y):=\frac{|y|}{|\phi'(\phi^{-1}(y^2))|}.
\end{equation}
The integrand in (\ref{ZNchangevar}) is almost a product of $N$ copies of $f(y)$. Multiply and divide the integrand by
$|y_N|=\sqrt{E-\sum_{i=1}^{N-1}y_i^2}$. Recall the following formula for integration over a sphere
\begin{equation*}
\begin{split}
&\int_{\s^{N-1}(E)}g(y_1,\dots,y_N)d\sigma^{(E^2)}\\
&\hspace{2cm}=\sum_{\epsilon=+,-}
\int_{\sum_{i=1}^{N-1}y_i^2 \leq E^2}g(y_1,\dots,\epsilon y_N)\frac{E dy_1\dots dy_{N-1}}{\sqrt{E^2-\sum_{i=1}^{N-1}}y_i^2}.
\end{split}
\end{equation*}
We now get
\begin{equation}\label{ZNkac}
Z_\phi(\sqrt{E})= \frac{2^{N-1}}{\sqrt{E}} \int_{\s^{N-1}(\sqrt{E})} \prod_{i=1}^{N}f(y_i) d\sigma^{(E)}.
\end{equation}

Having $Z_\phi$ given by (\ref{ZNkac}) is convenient in sense that, in \cite{kac11}, Kac determined the asymptotic behaviour of $Z_\phi(\sqrt{N})$ for large $N$ by using the \textit{saddle point} method (see e.g. \cite{rvg}). For completeness, we present each step of the result with rigorous justification with $f(y)$ given by (\ref{functioncy}). A short description of the saddle point method is given in the Appendix.

We start by computing the Laplace transform of $E \mapsto Z_\phi(\sqrt{E})$. The Laplace transform of $Z_\phi(\sqrt{E})$ is defined provided that $Z_\phi(\sqrt{E})$ grows at most exponentially. Since the behaviour of $Z_\phi(\sqrt{E})$ depends on the function $f(y)$ defined by (\ref{functioncy}) we assume that $\phi(y)$ is such that
\begin{equation}\label{expcondf}
f(y)\leq Ke^{b y^2},
\end{equation}
for some $K\geq 0$ and $b>0$. This condition ensures that $Z_\phi(\sqrt{E})$ grows at most exponentially.

Taking the Laplace transform of $Z_\phi(\sqrt{E})$, making the change of variable $r=\sqrt{E}$,
we have, for $w\in \C$  where $\Re(w)>b$
\begin{equation*}
\int_{0}^{\infty} e^{-wE}Z_\phi(\sqrt{E})dE = 2\int_{0}^{\infty} e^{-wr^2}r Z_\phi(r)dr.
\end{equation*}
Using (\ref{ZNkac}), the last equality equals
\begin{equation*}
2^{N} \left ( \int_{-\infty}^{\infty} e^{-wy^2}f(y) dy\right)^N.
\end{equation*}
From condition (\ref{expcondf}) and since $\Re(w)>b$ it follows that
\begin{equation}\label{defofPhi}
\Phi(w):=\int_{-\infty}^{\infty} e^{-wy^2}f(y) dy
\end{equation}
is an analytic function of $w$ for $\Re(w)>b$. By applying the inverse of the Laplace transform, we get
\begin{equation*}
Z_\phi(\sqrt{E})=\frac{2^{N}}{2\pi i} \int_{\gamma-i\infty}^{\gamma+i\infty}
 e^{zE}\left(\int_{-\infty}^{\infty}e^{-zy^2}f(y)dy \right)^N dz,
\end{equation*}
where $\gamma>b$ is such that the line $\gamma=\Re(z)$ lies in the half-plane where $\Phi(z)$ is analytic. Replacing $E$ with $N$, we get
\begin{equation*}
Z_\phi(\sqrt{N})=\frac{2^{N-1}}{\pi i} \int_{\gamma-i\infty}^{\gamma+i\infty}
 \left( e^{z} \int_{-\infty}^{\infty}e^{-zy^2}f(y)dy \right)^N dz.
\end{equation*}
By comparing with the saddle point integral (\ref{saddlepint}) in the Appendix, we set
\begin{equation}\label{defofS}
q(z)=1 \hspace{0.1in}\text{and} \hspace{0.1in} S(z)=z+\log \Phi(z),
\end{equation}
which is well defined because $\Phi(z)$ is nonzero on the line $\gamma=\Re(z)$. We now have
\begin{equation}\label{ZNsaddle}
Z_\phi(\sqrt{N})= \frac{2^{N-1}}{\pi i} \int_{\gamma-i\infty}^{\gamma+i\infty}e^{NS(z)}dz.
\end{equation}
The asymptotic behaviour of $Z_\phi(\sqrt{N})$ for large $N$ is determined by the saddle points of $S(z)$. The next lemma concerns the saddle points of S(z).
\begin{lemma}\label{lemmaZN}
Assume that there exists a $\gamma>0$ such that
\begin{equation}\label{assump0}
\int_{-\infty}^{\infty} e^{-\gamma y^2}f(y) dy < \infty.
\end{equation}
Moreover, assume that
\begin{equation}\label{assump1}
\int_{-\infty}^{\infty}(1-y^2)f(y)dy<0,
\end{equation}
and
\begin{equation}\label{assump2}
\int_{|y|\leq 1}(1-y^2)f(y)dy > 0.
\end{equation}
Let S(z) be given by (\ref{defofS}) with $\Phi(z)$ by (\ref{defofPhi}). For $\Re(z)\geq \gamma$, the function $S(z)$ is analytic and there exists a unique saddle point $z_0$ to $S(z)$ such that $z_0$ is real, $z_0 \geq \gamma$ and $S''(z_0)>0$. Moreover
\begin{equation}\label{asympZN}
Z_\phi(\sqrt{N})\sim \frac{2^{N-1}e^{Nz_0}}{\sqrt{NS''(z_0)}} \left ( \int_{-\infty}^{\infty} e^{-z_0 y^2} f(y) dy \right )^N.
\end{equation}
\end{lemma}
\begin{proof}
Note that, for $z=\xi+i\eta$, for all $\xi$
\begin{equation*}
\arg\max_\eta |e^{\xi+i\eta}\Phi(\xi+i\eta)|=\{0\}.
\end{equation*}
Hence, we only need to find saddle points on the real line.\\ \\
\textbf{Claim 1}: The function $S(z)$ has a unique saddle point $z_0$ where $z_0\geq \gamma$.
\begin{proof}
For $\Re(z)\geq \gamma$, we have
\begin{equation}\label{eqSprim}
S'(z)=1-\frac{\int_{-\infty}^{\infty}y^2 e^{-zy^2}f(y)dy}{\int_{-\infty}^{\infty}e^{-zy^2}f(y)dy}.
\end{equation}
For $z=\xi+i0$ where $\xi \geq \gamma$, the equation $S'(\xi)=0$ is equivalent to
\begin{equation*}
\int_{-\infty}^{\infty}e^{-\xi y^2}(1-y^2)f(y)dy=0.
\end{equation*}
Multiplying the last equality by $e^{\xi}$, we see that, $S'(\xi)=0$ is equivalent to
\begin{equation*}
\int_{-\infty}^{\infty}e^{-\xi (y^2-1)}(1-y^2)f(y)dy=0.
\end{equation*}
Let
\begin{equation*}
A(\xi)=\int_{-\infty}^{\infty}e^{-\xi (y^2-1)}(1-y^2)f(y)dy.
\end{equation*}
Since
\begin{equation*}
A'(\xi)=\int_{-\infty}^{\infty}e^{-\xi (y^2-1)}(1-y^2)^2f(y)dy>0,
\end{equation*}
it follows that $A(\xi)$ is an increasing function of $\xi$. Moreover, by (\ref{assump1}) we have
\begin{equation*}
A(0)<0.
\end{equation*}
Note that
\begin{equation*}
\lim_{\xi \rightarrow \infty}\int_{|y|> 1}e^{-\xi (y^2-1)}(1-y^2)f(y)dy=0.
\end{equation*}
Hence
\begin{eqnarray*}
&&\int_{-\infty}^{\infty}e^{-\xi (y^2-1)}(1-y^2)f(y)dy \sim \int_{|y|\leq 1}e^{-\xi (y^2-1)}(1-y^2)f(y)dy \nonumber\\
&&=\int_{|y|\leq 1}e^{\xi (1-y^2)}(1-y^2)f(y)dy \nonumber.
\end{eqnarray*}
The last integral goes to infinity by (\ref{assump2}) as $\xi\rightarrow \infty$. Hence, there exists a unique
$z_0 \geq \gamma$ such that $S'(z_0)=0$.
\end{proof}
\textbf{Claim 2}: The second derivative of $S(z)$ at $z_0$ is positive.
\begin{proof}
We have
\begin{equation*}
S''(z)=\frac{\Phi''(z)}{\Phi(z)}-\frac{\Phi'(z)^2}{\Phi(z)^2}.
\end{equation*}
For $z=z_0$, using the Jensen inequality, we get
\begin{equation*}
\frac{\Phi''(z_0)}{\Phi(z_0)}=\frac{\int_{-\infty}^{\infty}y^4 e^{-z_0 y^2}f(y)dy}{\int_{-\infty}^{\infty} e^{-z_0 y^2}f(y)dy} >
\left ( \frac{\int_{-\infty}^{\infty}y^2 e^{-z_0 y^2}f(y)dy}{\int_{-\infty}^{\infty} e^{-z_0 y^2}f(y)dy}\right )^2=\frac{\Phi'(z_0)^2}{\Phi(z_0)^2}.
\end{equation*}
This proves the claim.
\end{proof}
We now turn to the proof of (\ref{asympZN}). We can write
\begin{equation*}
\Phi(z)= \int_{-\infty}^{\infty} e^{-z y^2}f(y)dy = \int_{0}^{\infty} e^{-z y^2}(f(y)+f(-y))dy.
\end{equation*}
By a change of variables, the last integral equals
\begin{equation*}
\int_{0}^{\infty}e^{-z y} \frac{f(\sqrt{y})+f(-\sqrt{y})}{2\sqrt{y}}dy.
\end{equation*}
Let $z=\xi+i\eta$, with $\xi \geq \gamma$. The last integral can be written as
\begin{equation*}
\int_{-\infty}^{\infty}e^{-\xi y} \frac{f(\sqrt{y})+f(-\sqrt{y})}{2\sqrt{y}} \mathbbm{1}_{y\geq 0}\hspace{0.05in} e^{-i \eta y}dy.
\end{equation*}
The last integral is the Fourier transform of the function $\tilde{\Phi}_{\xi}$ at the point $\eta$, where
\begin{equation*}
\tilde{\Phi}_{\xi}(y)=e^{-\xi y} \frac{f(\sqrt{y})+f(-\sqrt{y})}{2\sqrt{y}} \mathbbm{1}_{y\geq 0}.
\end{equation*}
For $\xi=z_0$, it follows that $\tilde{\Phi}_{\xi}\in L^1(\R)$, and $|\mathcal{F}(\tilde{\Phi}_{\xi})(\eta)|<\mathcal{F}(\tilde{\Phi}_{\xi})(0)$. Moreover, By the Riemann Lebesgue lemma, it follows that
\begin{equation*}
|\mathcal{F}(\tilde{\Phi}_{\xi})(\eta)|\rightarrow 0 \hspace{0.1in} \hbox{when} \hspace{0.1in} |\eta| \rightarrow \infty.
\end{equation*}
Let $C_{+T}$ and $C_{-T}$ be the curves in the complex plane given $[\gamma+ iT, z_0+iT]$ and $[\gamma-iT,z_0-iT]$, respectively.

We have
\begin{eqnarray*}
  \left| \int_{C_{\pm T}} e^{N S(z)}dz \right| &=&  \left| \int_{\gamma}^{z_0}e^{N(\xi+ iT)} \mathcal{F}(\tilde{\Phi}_{\xi})(T)^N d\xi \right|  \nonumber \\
   &\leq& \int_{\gamma}^{z_0} e^{N\xi} |\mathcal{F}(\tilde{\Phi}_{\xi})(T)|^N d\xi \nonumber \\
   &\leq& |z_0-\gamma|\max_{\xi \in [\gamma,z_0]} e^{N \xi} |\mathcal{F}(\tilde{\Phi}_{\xi})(T)|^N \rightarrow 0,
   \hspace{0.02in}\hbox{as}\hspace{0.01in} |T|\rightarrow \infty. \nonumber \\
\end{eqnarray*}
By Cauchy's theorem, we can deform the contour in (\ref{ZNsaddle}) to a contour passing through the saddle point $z_0$.
Hence
\begin{equation*}
Z_\phi(\sqrt{N})=\frac{2^{N-1}}{\pi i} \int_{z_0-i\infty}^{z_0+i\infty}e^{NS(z)}dz.
\end{equation*}
By the saddle point method, we obtain
\begin{eqnarray}
Z_\phi(\sqrt{N})&=&\frac{2^{N-1}}{\pi i} \sqrt{\frac{-2\pi}{NS''(z_0)}}  \left(1 + \Ordo \left(\frac{1}{N}\right )\right) e^{NS(z_0)} \nonumber \\
&\sim& \frac{2^{N}e^{Nz_0}}{\sqrt{NS''(z_0)}} \left ( \int_{-\infty}^{\infty} e^{-z_0 y^2} f(y) dy \right )^N.
\end{eqnarray}
This finishes the proof of the lemma.
\end{proof}
The main theorem of this section is:
\begin{theorem}\label{theoremchaos}
Let $f(y)$ defined by (\ref{functioncy}) satisfy the conditions
(\ref{expcondf}), (\ref{assump1}) and (\ref{assump2}). Then, the family of
uniform densities on~$\Omega^{N-1}(\sqrt{N})$ with respect to the
microcanonical measure is $Ce^{-z_0 \phi(v)}-$ chaotic. The positive constant
$z_0$ is the unique real solution to the following equation
\begin{equation*}
\int_{-\infty}^{\infty}(1-\phi(v))e^{-z_0\phi(v)}dv=0,
\end{equation*}
and $C$ is a normalisation constant given by
\begin{equation*}
C=\displaystyle\frac{1}{\int_{-\infty}^{\infty}e^{-z_0\phi(v)}dv}.
\end{equation*}
\end{theorem}
\begin{proof}
Let $\varphi$ be a bounded continuous function on $\R^k$. For each fixed $k\in \N$, a modification of Lemma \ref{lemmaZN} shows that, for large $N$,
\begin{equation*}
\begin{split}
  \displaystyle\sum_{\epsilon=+,-}\int_{\sum_{i=1}^{N-1}\phi(v_i)\leq E}&\varphi(v_1,\dots,v_k)\frac{1}{|\frac{\partial H}{\partial v_N}|_\epsilon}dv_1 \dots dv_{N-1}\\
   &\sim \displaystyle \int_{\R^k}\varphi(\phi^{-1}(y_1^2),\dots,\phi^{-1}(y_k^2))\prod_{i=1}^{k}e^{-z_0 y_i^2}f(y_i)dy_1 \dots d y_k\\
   &\hspace{0.7in}\times \frac{2^Ne^{(N-k)z_0}}{\sqrt{NS''(z_0)}}\left ( \int_{-\infty}^{\infty} e^{-z_0 y^2} f(y) dy \right )^{N-k}.
\end{split}
\end{equation*}
Using Lemma \ref{lemmaZN} again and making the change of variable
$y_i^2=\phi(v_i)$, where $dv_i=f(y_i)dy_i$, $i=1,\dots,k$ leads to
\begin{equation*}
\begin{split}
\lim_{N \rightarrow \infty }&\displaystyle \frac{\int_{\Omega^{N-1}(\sqrt{N})}\varphi(v_1,\dots,v_k)
\hspace{0.01in}\hspace{0.01in} d\eta^{(N)}}{Z_{\phi}(\sqrt{N})}\\
&=\frac{\displaystyle \int_{\R^k}\varphi(\phi^{-1}(y_1^2),\dots,\phi^{-1}(y_k^2))\prod_{i=1}^{k}e^{-z_0 y_i^2}f(y_i)dy_1 \dots d y_k}
{\displaystyle \int_{\R^k}\prod_{i=1}^{k}e^{-z_0 y_i^2}f(y_i)dy_1 \dots d y_k}\\
&=\frac{\displaystyle \int_{\R^k}\varphi(v_1,\dots,v_k)\prod_{i=1}^{k}e^{-z_0 \phi(v_i)} dv_1 \dots dv_k}
{\displaystyle \int_{\R^k}\prod_{i=1}^{k}e^{-z_0 \phi(v_i)} dv_1 \dots dv_k}.
\end{split}
\end{equation*}
This is what we wanted to prove.
\end{proof}
We now consider two examples where Theorem \ref{theoremchaos} applies :
\begin{example}
In the classical Kac model $\phi(v)=v^2$. We get that $f(y)=\frac{1}{2}$ and thus, the conditions (\ref{assump0}), (\ref{assump1}) and (\ref{assump2}) are fulfilled.
To determine $z_0$ we need solve the equation
\begin{equation*}
\int_{-\infty}^{\infty}(1-v^2)e^{-z_0 v^2}dv=0.
\end{equation*}
Direct calculations show that $z_0=\frac{1}{2}$ is the unique real solution.
Therefore the uniform density on
($\Omega^{N-1}(\sqrt{N})=\s^{N-1}(\sqrt{N})$) with respect to the
microcanonical measure is $\frac{e^{-{\frac{1}{2}}v^2}}{\sqrt{2\pi}}-$chaotic.
This has already been discussed in Example \ref{exchaoticsph}. Recall that, on
$\s^{N-1}(\sqrt{N})$, the micro\-canonical measure up to a constant is equal to
the surface measure
\end{example}
\begin{example}
For a relativistic energy function $\phi(v)=\sqrt{v^2+1}-1$, it follows that
\begin{equation*}
f(y)=\frac{(y^2+1)|y|}{\sqrt{(y^2+1)^2-1}}.
\end{equation*}
It is easy to check that the conditions (\ref{assump0}), (\ref{assump1}) and
(\ref{assump2}) are satisfied. To find $z_0$ we solve the equation
\begin{equation*}
\int_{-\infty}^{\infty}(1-(\sqrt{v^2+1}-1))e^{-z_0(\sqrt{v^2+1}-1)}dv=0.
\end{equation*}
By using numerical integration, we get that $z_0 \approx 0.734641$. Hence, the
family of uniform densities on $\Omega^{N-1}(\sqrt{N})$ with
$\phi(v)=\sqrt{v^2+1}-1$ in (\ref{spaceOmega}) is
\\
$Ce^{-z_0(\sqrt{v^2+1}-1)}$- chaotic, where $C \approx 4.082$.
\end{example}
We end this section by a comparison between our result and the result of
Sznitman.The setting there is slightly different, and as commented by Sznitman it is not directly applicable for the case $h=v^2$ or the more general situation here.
To see the relation between two results, assume that an individual velocity is distributed according the Gibbs
measure $e^{-\beta \phi(v)}/Z$, where $Z$ is normalization constant and
$\beta>0$. Then the random variable $h=\phi(v)$ that is the energy of a
particle, has a distribution $\mu(dh)=f(h)dh$, where
\begin{equation}
  f(h)=\frac{e^{-\beta h}}{\left|\phi'(\phi^{-1}(h))\right|}.
\end{equation}

To apply Sznitman's result, we need to show that $f$ satisfies condition
(\ref{conditionflambda}). This is a strong integrability condition on $f$ and
depends on the choice of $\phi$. Unfortunately, even for the most classical
case, the one studied by Kac, where $h=v^2$, condition (\ref{conditionflambda})
is not fulfilled since $f'(h)e^{-\beta h}$ is not integrable at $h=0$. It is a
technical condition needed to control the Fourier transform of $f$, and it
could probably be relaxed, but it means that Sznitman's result cannot be
directly applied to our case.

\section{Particle dynamics and master equation for general energy functions}

In this section we follow \cite{cdw} and \cite{kac11} to introduce dynamics
between particles having the energy given by the function $\phi$ and obtain the
corresponding master equation. By similar arguments as in \cite{kac11}, we will
see that the master equation propagates chaos.

To introduce dynamics between the particles, let the \textit {master vector}\\
$\textbf{V}~=(v_1,\dots,v_N)\in \Omega^{N-1}(\sqrt{N})$. The master vector
$\textbf{V}$ makes a jump on $\Omega^{N-1}(\sqrt{N})$ according  the following
steps:
\begin{enumerate}
  \item Pick a pair $(i,j)$, $i<j$ according to the uniform distribution
  $$P_{ij}=\frac{2}{N(N-1)}.$$
  \item The pair of velocities $(v_i,v_j)$ satisfies
  $$\phi(v_i)+\phi(v_j)=h,\hspace{0.2in} h>0.$$
  Let
  $$y_i=\mbox{sign}(v_i)\sqrt{\phi(v_i)} \hspace{0.2in}\text{and}\hspace{0.2in}
      y_j=\mbox{sign}(v_j)\sqrt{\phi(v_j)}.$$ Then $(y_i,y_j)$ is a point on
  the circle, and as in the original Kac model,  the collision may be
  performed there.

  Pick an angle $\theta$ uniformly on $(0,2\pi]$ and let
  $$y_i'(\theta)=y_i\cos\theta+y_j\sin\theta \hspace{0.1in}\text{and}
      \hspace{0.1in} y_j'(\theta)=-y_i\sin\theta+y_i\cos\theta.$$ The pair
  $(v_i,v_j)$ is transformed on $\Omega^1(\sqrt{h})$ to
  $(v_i'(\theta),v_j'(\theta))$ according to
  \begin{align}\label{cols}
  \nonumber v_i'(\theta)&=\mbox{sign}(y_i'(\theta))\phi^{-1}\left(y_i'(\theta)^2\right),\\
   v_j'(\theta)&=\mbox{sign}(y_j'(\theta))\phi^{-1}\left(y_j'(\theta)^2 \right),
  \end{align}
  where
  \begin{equation*}
  \mbox{sign}(v) = \begin{cases} 1, & \mbox{if } v \geq 0 \\ -1, & \mbox{if } v<0. \end{cases}
  \end{equation*}
  Note that
  $$\phi(v_i'(\theta))+\phi(v_j'(\theta))=\phi(v_i)+ \phi(v_j).$$
  \item Update the master vector $\textbf{V}$ and denote the new master
      vector by $T_{i,j}(\theta)\textbf{V}$. By step $2$, it follows that
      $T_{i,j}(\theta)\textbf{V}\in \Omega^{N-1}(\sqrt{N})$. Repeat step $1$,
      $2$ and $3$.
\end{enumerate}
This is only a generalization of the dynamics in \cite{kac11} where $\phi(v)=v^2$.

The steps above describe a random walk on $\Omega^{N-1}(\sqrt{N})$.  As in
\cite{cdw}, we define its Markov transition operator $Q_{\phi}$. If $V_k$ is
the state of the particles after the k-th step of the walk, for a continuous
function $\varphi$ on $\Omega^{N-1}(\sqrt{N})$, the operator $Q_{\phi}$ is
defined by
\begin{equation*}
Q_{\phi}\varphi(y)=\E \left[\varphi(V_{k+1})| V_k=y\right].
\end{equation*}
Writing out the expectation above, we get
\begin{equation}\label{defofQ}
Q_{\phi}\varphi(\textbf{V})=\frac{2}{N(N-1)}\sum_{i<j}
     \int_0^{2\pi}\varphi(T_{i,j}(\theta)\textbf{V})\frac{d\theta}{2\pi}.
\end{equation}
If $F_k$ is probability density of $V_k$ with respect to the  microcanonical
measure $\eta^{(N)}$ on $\Omega^{N-1}(\sqrt{N})$, we have
\begin{equation*}
\begin{split}
\int_{\Omega^{N-1}(\sqrt{N})}&\varphi F_{k+1} d \eta^{(N)}=\E [\varphi(V_{k+1})]=\E \left[ \E[\varphi(V_{k+1})| V_k]\right]\\
&=\int_{\Omega^{N-1}(\sqrt{N})}Q_{\phi}\varphi F_k d \eta^{(N)}.
\end{split}
\end{equation*}
By definition, the microcanonical measure is invariant under the transformation
$\textbf{V} \rightarrow T_{i,j}(\theta)\textbf{V}$. It follows that $Q_{\phi}$
is self adjoint and
\begin{equation*}
F_{k+1}=Q_{\phi}F_k.
\end{equation*}

So far, the process defined above is discrete in time. To obtain a time
continuous process, we let the master vector $\textbf{V}$ be a function of
time, and the times between the jumps (collisions) exponentially distributed.
In this way, if $F_N(\textbf{V},0)$ is the probability distribution of the $N$
particles on $\Omega^{N-1}(\sqrt{N})$ at time $0$, the time evolution of
$F_N(\textbf{V},t)$ is described by the following master equation
\begin{equation}\label{mastereqH}
\frac{\partial F_N(\textbf{V},t) }{\partial t}=\mathcal{K}_{\phi} F_N(\textbf{V},t),
\end{equation}
where
\begin{equation}\label{defofK}
\mathcal{K}_{\phi}=N[Q_{\phi}-I]
\end{equation}
and $I$ is the identity operator. A more complete discussion of master
equations of this kind may be found in \cite{cdw}.

Note that, for $\phi(v)=v^2$, the collision operator $\mathcal{K}_{\phi}$ is
the same as the collision operator $\mathcal{K}$ in the Kac model. We have
propagation of chaos for the master equation (\ref{mastereqH}):

\begin{theorem}
Assume that the family of initial densities $\{F_N(\textbf{V},0)\}_{N\in \N}$ on $\Omega^{N-1}(\sqrt{N})$  is $f(v,0)-$chaotic. Then, the family of densities $\{F_N(\textbf{V},t)\}_{N\in \N}$, where $F_N(\textbf{V},t)$ is the solution to (\ref{mastereqH}) is $f(v,t)-$chaotic. Moreover, the density $f(v,t)$ satisfies the following equation
\begin{equation}
\frac{\partial }{\partial t}f(v,t)= \int_\R \int_{0}^{2\pi}
\left [f(v(\theta),t)f(w(\theta),t)-f(v,t)f(w,t) \right]\ \frac{d\theta}{2\pi} dw,
\end{equation}
with $f(v,0)=f_0(v)$ and $v(\theta)$, $w(\theta)$ given by (\ref{cols}).
\end{theorem}
\begin{proof}
The proof follows by the same arguments as in \cite{kac11}. For a more detailed
proof where propagation of chaos is shown for more general master equations, we
refer to \cite{cdw}.
\end{proof}

\section{Chaotic measures in higher dimensions}
In Section~$2$ we proved that the uniform density with respect to the microcanonical
measure on $\Omega^{N-1}(\sqrt{N})$ is $Ce^{-z_0\phi(v)}$-chaotic, $v\in \R$.
The goal of this section is to generalize this to the case $v\in \R^2$ where
now both the energy and momentum are conserved; the generalization to $\R^d$,
$d>2$ may be treated in the same way. The calculations here are formal. For
$p\in \R^2$, define
\begin{equation}
\Gamma^N(\sqrt{E},p)=
\left \{(v_1,\dots,v_N)\in \R^{2N}\;\Big | \;\sum_{i=1}^{N}\phi(v_i)=2E,\; \sum_{i=1}^{N}v_i=p \right\}.
\end{equation}
We assume that $E$ and $p$ are chosen such that $\Gamma^N(\sqrt{E},p)$ is
non-empty. The classical case when $\phi(v)=|v|^2$ has been thoroughly
investigated in \cite{ck}.

For $p=(p_1,p_2)$ and $v_i=(v_{i1},v_{i2})$, $i=1,\dots N$, a measure
$\mu_{E,p_1,p_2}$ concentrated on $\Gamma^N(\sqrt{E},p)$ is defined by
\begin{equation}
\mu_{E,p_1,p_2}=\delta (2E - \sum_{i=1}^{N}\phi(v_i))\delta (p_1 - \sum_{i=1}^{N}v_{i1})
\delta (p_2 - \sum_{i=1}^{N}v_{i2}).
\end{equation}
The product of the Dirac measures is well defined since the hyper surfaces
defined by setting the arguments of the Dirac measures to zero are mutually
transversal. Let
\begin{equation}
Z(E,p_1,p_2)=\displaystyle\int_{\R^{2N}}\delta (2E - \sum_{i=1}^{N}\phi(v_i))
\delta (p_1 - \sum_{i=1}^{N}v_{i1})
\delta (p_2 - \sum_{i=1}^{N}v_{i2}) dv_1 \dots dv_N.
\end{equation}

As in Section $2$, we need to determine the asymptotic behaviour of
$Z(E,p_1,p_2)$ with $E=N$ for large $N$. We note that in the case of
relativistic collisions, with $\phi(v)=\sqrt{|v|^2+1}-1$, the measure
$\mu_{E,p_1,p_2}$ is Lorentz invariant (see \cite{robert}).

In the sense of distributions, the $\delta$ function is the inverse Fourier
transform of the function $1$ and formally can be written as
\begin{equation}
\delta(x)=\frac{1}{2\pi}\int_\R e^{i x \xi}d\xi.
\end{equation}
For $z=(z_1,z_2,z_3)$, we can formally write

\begin{equation*}
\begin{split}
&Z(E,p_1,p_2)\\
&=\frac{1}{(2\pi)^3}\displaystyle \int_{\R^{2N}}\int_{\R^3}e^{i(2E - \sum_{i=1}^{N}\phi(v_i))z_3}
e^{i(p_1 - \sum_{i=1}^{N}v_{i1})z_1}e^{i(p_2 - \sum_{i=1}^{N}v_{i2})z_2} dz_1 dz_2 dz_3 \ dV,
\end{split}
\end{equation*}
where $dV =dv_1 \dots dv_N$. With $E=N$, the last equality is
\begin{equation*}
\begin{split}
Z&(N,p_1,p_2)=\\
&\frac{1}{(2\pi)^3}\displaystyle \int_{R^3}e^{ip_1 z_1+ip_2 z_2}
\left(e^{2i z_3}\int_{\R^2} e^{-i\phi(v_1) z_3-iv_{11} z_1-iv_{12} z_2} dv_{11}dv_{12}\right)^N dz_1 dz_2 dz_3.
\end{split}
\end{equation*}
Note that here $p_1$ and $p_2$ are assumed to be independent of $N$. A natural
and straightforward variation is to replace $p_1$ by $Np_1$ and $p_2$ by$Np_2$.\\

Let
\begin{equation}\label{Sdim3}
\begin{split}
&q(z_1 , z_2 , z_3)=e^{ip_1 z_1 +ip_2 z_2},\\
&S(z_1 , z_2 , z_3)=2i z_3+ \log \int_{\R^2}e^{-i\phi(v_1) z_3 - iv_{11} z_1-iv_{12} z_2} dv_{11}dv_{12},
\end{split}
\end{equation}
where $q:\C^3\rightarrow \R$, $S:\C^3\rightarrow \R$. With these notations,  we
can write
\begin{equation}\label{Zsadelmult}
Z(N,p_1,p_2)=\frac{1}{(2\pi)^3}\displaystyle \int_{\R^3}q(z_1 , z_2 , z_3)e^{NS(z_1 , z_2 , z_3)}dz_1 dz_2 dz_3.
\end{equation}
The right hand side of the last equality is a saddle point integral in
dimension~3. The asymptotic behaviour of $Z(N,p_1,p_2)$ for large $N$ is
determined by the saddle points of $S(z_1 , z_2 , z_3)$, i.e., points $(\bar
z_1, \bar z_2, \bar z_3)$ such that
\begin{equation}\label{sadelcond}
\nabla S(\bar z_1,\bar z_2,\bar z_3)=0.
\end{equation}
Using (\ref{Sdim3}), we need to solve the following system of equations:
\begin{eqnarray}
  2i-\frac{i\int_{\R^2}\phi(v_1)e^{-i\phi(v_1)\bar z_3-iv_{11} \bar z_1-iv_{12}\bar z_2} dv_{11}dv_{12}}{\int_{\R^2}e^{-i\phi(v_1)\bar z_3-iv_{11}\bar z_1-iv_{12}\bar z_2} dv_{11}dv_{12}} &=& 0, \\
\label{eq2}  -\frac{i\int_{\R^2}v_{11}e^{-i\phi(v_1)\bar z_3-iv_{11}\bar z_1-iv_{12}\bar z_2} dv_{11}dv_{12}}{\int_{\R^2}e^{-i\phi(v_1)\bar z_3-iv_{11}\bar z_1-iv_{12}\bar z_2} dv_{11}dv_{12}} &=& 0, \\
\label{eq3}  -\frac{i\int_{\R^2}v_{12}e^{-i\phi(v_1)\bar z_3-iv_{11}\bar z_1-iv_{12}\bar z_2} dv_{11}dv_{12}}{\int_{\R^2}e^{-i\phi(v_1)\bar z_3 -iv_{11}\bar z_1-iv_{12}\bar z_2} dv_{11}dv_{12}} &=& 0.
\end{eqnarray}
Since $\phi$ is even, it follows that $\bar z_1,\bar z_2=0$ are the unique solutions to equations (\ref{eq2}) and (\ref{eq3}). We can now obtain $\bar z_3$ by solving the following equation:
\begin{equation}\label{eqxi}
\int_{\R^2}(2-\phi(v_1))e^{-i\phi(v_1)\bar z_3} dv_1=0.
\end{equation}
Assuming that $(0,0, \bar z_3)$ is the unique solution to (\ref{sadelcond}), and that we can deform the integration domain in (\ref{Zsadelmult}) to contain $(0,0, \bar z_3)$, we find that
\begin{equation}
\begin{split}
Z(N,p_1,p_2)& \sim \frac{1}{(2\pi)^3}\left(\frac{2\pi}{N} \right)^{3/2}\frac{1}{(\mbox{det} S''((0,0, \bar z_3)))^{1/2}}
e^{NS(0,0, \bar z_3)}q(0,0, \bar z_3)\\
&=\frac{1}{(2\pi N)^{3/2}}\frac{1}{(\mbox{det} S''((0,0, \bar z_3)))^{1/2}} e^{2Ni\bar z_3}
\left(\int_{\R^2}e^{-i\phi(v_1)\bar z_3}dv_1\right)^N.
\end{split}
\end{equation}
By the discussion in Section $2$, this procedure shows formally that the uniform distributions on $\Gamma^N(\sqrt{N},p)$ with respect to the measure $\mu_{N,p_1,p_2}$ is $Ce^{-z_0\phi(v)}$-chaotic, where $z_0=i\bar z_3$ and $\bar z_3$ the unique solution to (\ref{eqxi}) and
$$C=\displaystyle\frac{1}{\int_{\R^2}e^{-z_0 \phi(v)}dv}.$$

Sznitman's method, referred to in Section~\ref{sec:intro} and
the end of Section~\ref{sec:chaotic}, is not restricted to the one-dimensional
setting, but could formally be used also here: with $N$ spatial dimensions, we
would have $d=N+1$, and $h$ in Equations~(\ref{conditionflambda}) and
(\ref{eq:sz14}) would be a generalized four-momentum, $h=(v_1,v_2,...,v_N,
\phi(v_1,...,v_N) )$. However, the same difficulty as in the one-dimensional
case would appear here, and a modification of Sznitman's argument would be
needed, in the same way.


\section{Appendix}
The saddle point method is used to determine the asymptotic behaviour of integrals depending on a parameter. For a detailed description we refer to \cite{rvg}. Without proof we only state below the saddle point method which is concerned with this paper.\\
\textbf{One-dimensional saddle point method}\\
Let $\gamma$ be a contour in the complex plane. Assume that $q$ and $S$ are analytic functions in a neighborhood of the contour $\gamma$. Consider the following integral
\begin{equation}\label{saddlepint}
F(\lambda) =\int_{\gamma} q(z)e^{\lambda S(z)}dz.
\end{equation}
A point $z_0\in \C$ is called a simple \textit{saddle point } of the function $S:\C\rightarrow \C$ if $S'(z_0)=0$ and $S''(z_0)\neq 0$.
Assume that $z_0 \in \gamma$ is the unique simple saddle point of $S$. Then as $\lambda\rightarrow \infty$
\begin{equation}
F(\lambda)=\sqrt{\frac{-2\pi}{\lambda S''(z_0)}}e^{\lambda S(z_0)}\left (q(z_0)+\Ordo(\frac{1}{\lambda}) \right ).
\end{equation}
If there are more than one saddle point, $F$ will be expressed as a sum over these points.\\
\textbf{Many-dimensional saddle point method }\\
Let $\gamma$ be an $N$-dimensional smooth compact manifold. Consider the following integral
\begin{equation}
F(\lambda) =\int_{\gamma} q(z)e^{\lambda S(z)}dz.
\end{equation}
where $z=(z_1,\dots,z_N)\in \C^N$ and the functions $q(z)$ and $S(z)$ are assumed to be analytic in a domain $D$ containing $\gamma$. A point $z_0$ is called a simple saddle point of $S(z)$ if $\nabla S(z_0)=0$ and
$det S''(z_0)\neq 0$. Assume that $z_0 \in \gamma$ is the unique simple saddle point of $S$. Then as $\lambda\rightarrow \infty$
\begin{equation}
F(\lambda)=\left(\frac{2\pi}{\lambda}\right)^{N/2}\frac{1}{(\mbox{det} S''(z_0))^{1/2}}e^{\lambda S(z_0)}\left (q(z_0)+\Ordo(\frac{1}{\lambda})
\right ).
\end{equation}
If there are more than one saddle point, $F$ is be expressed as a sum over these points.
\section*{Acknowledgments}
The authors would like to thank one of the referees for several suggestions.
D.M. acknowledges support by the Swedish Science Council. B.W. acknowledges
support by the Swedish Science Council, the Knut and Alice Wallenberg
foundation and the Swedish Foundation for Strategic Research.

\end{document}